\documentclass{article}
\usepackage[utf8]{inputenc}

\usepackage{amsmath}
\usepackage{amssymb}
\usepackage{amsthm}
\usepackage{tikz}
\usepackage{graphicx}
\usepackage{float}
\usepackage{algorithm}
\usepackage[noend]{algpseudocode}
\usepackage{tabularx}
\usepackage[margin=1in]{geometry}
\usepackage{arydshln}
\usepackage{multirow}
\usepackage{wrapfig}
\usepackage{xcolor}
\usepackage{xspace}
\usepackage{hyperref}
\usepackage{comment}
\usepackage{bbm}
\usepackage{todonotes}

\usetikzlibrary{positioning,calc,shapes,arrows}
\usetikzlibrary{backgrounds}
\usetikzlibrary{matrix,shadows,arrows} 
\usetikzlibrary{decorations.pathreplacing}

\def\etal.{et\penalty50\ al.}
\theoremstyle{plain}
\newtheorem{theorem}{Theorem}[section]

\newtheorem{fact}[theorem]{Fact}

\theoremstyle{definition}

\theoremstyle{remark}
\newtheorem{question}{Question}[section]

\newtheorem{openproblem}[question]{Open Problem}

\theoremstyle{plain}
\newtheorem*{theorem*}{Theorem}

\DeclareMathOperator*{\argmin}{arg\,min}

\newcommand*\samethanks[1][\value{footnote}]{\footnotemark[#1]}
\newcommand{\infdiv}[2]{D(#1||#2)}

\algnewcommand{\LineComment}[1]{\State \(\triangleright\) #1}

\providecommand{\keywords}[1]
{
  \small	
  \textbf{\textit{Keywords---}} #1
}

\title{Advice Complexity of Online Non-Crossing Matching}	
\author{Ali Mohammad Lavasani\\	
Concordia University, CSSE\\	
\href{mailto:ali.mohammadlavasani@concordia.ca}{ali.mohammadlavasani@concordia.ca\thanks{Research is supported by NSERC.}}\and	
Denis Pankratov\\	
Concordia University, CSSE\\	
\href{mailto:denis.pankratov@concordia.ca}{denis.pankratov@concordia.ca}\samethanks}	
\date{\today}	
\begin{document}	
\maketitle

\begin{abstract}

    We study online matching in the Euclidean $2$-dimesional plane with the non-crossing constraint. The offline version was introduced by Atallah in 1985 and the online version was introduced and studied more recently by Bose et al. The input to the monochromatic non-crossing matching (MNM) problem consists of a sequence of points. 
    Upon arrival of a point, an algorithm can decide to match it with a previously unmatched point or leave it unmatched. The line segments corresponding to the edges in the matching should not cross each other, and the goal is to maximize the size of the matching.
    The decisions are irrevocable, and while an optimal offline solution always matches all the points, an online algorithm cannot match all the points in the worst case, unless it is given some additional information, i.e., advice. In the bichromatic  version (BNM), blue points are given in advance and the same number of red points arrive online. The goal is to maximize the number of red points matched to blue points without creating any crossings. 

    We show that the advice complexity of solving BNM  optimally on a circle (or, more generally, on inputs in a convex position) is tightly bounded by the logarithm  of  the $n^\text{th}$ Catalan  number  from  above  and  below.   This  result  corrects  the  previous  claim  of Bose et al. that the advice complexity is $\log(n!)$.  At the heart of the result is a connection between the non-crossing constraint in online inputs and the $231$-avoiding property of permutations of $n$ elements. We also show a lower bound of $n/3-1$ and an upper bound of $3n$ on the advice complexity for MNM on a plane.  This gives an exponential improvement over the previously best-known lower bound and an improvement in the constant of the leading term in the upper bound.  
    In addition, we establish a  lower bound of $\frac{\alpha}{2}\infdiv{\frac{2(1-\alpha)}{\alpha}}{1/4}n$  on the advice complexity for achieving competitive ratio $\alpha \in (16/17,1)$ for MNM on a circle where $\infdiv{p}{q}$ is the relative entropy between two Bernoulli random variables with parameters $p$ and $q$.
    Standard tools from advice complexity, such as partition trees and reductions from string guessing problems, do not seem to apply to MNM/BNM, so we have to design our lower bounds from first principles.
\end{abstract}


\section{Introduction}
Suppose that $2n$ points in general position are revealed one by one in the Euclidean plane. When and only when a point is revealed you have a choice of connecting it with a previously revealed (but yet unconnected) point by a straight-line segment. Each point can be connected to (alternatively, we say “matched to”) at most one other point and the connecting segments should not intersect each other: the connecting segments should form a non-crossing matching. Each decision on how to connect a point is irrevocable and the goal is to maximize the number of connected points.

It is easy to see that a perfect matching (matching involving all the points) of minimum total length\footnote{Total length of the matching is the sum of lengths of segments corresponding to edges in the matching.} is non-crossing. One can establish it by contradiction as follows. Suppose that a perfect matching of minimum total length contains two crossing segments. These crossing segments form diagonals of a quadrilateral and they can be exchanged for two sides of this quadrilateral giving another perfect matching of smaller total length. This contradicts the assumption of the initial perfect matching being of minimum total length. Thus, an optimal solution always matches all points.
However, this optimal solution requires the knowledge of the entire input sequence, and in our version of the problem the points arrive online and decisions are irrevocable. The question we are concerned with in this paper is how much additional information an online algorithm needs so that it can achieve optimal or near-optimal performance. We study two variants of the above problem: Online Monochromatic Non-Crossing Matching (or MNM, for short) where each point can be potentially matched with any other point, and Online Bichromatic Non-Crossing Matching (or BNM, for short) where $n$ of the points are colored blue and $n$ of the points are colored red, all the blue points arrive first, and points can be matched only if they have different colors. BNM can be thought of as a bipartite variant of the MNM in the geometric setting.

Matching problems in general abstract graphs have a long history dating back to K\"{o}nig’s theorem from 1931 and Hall’s marriage theorem from 1935. The offline setting of matching algorithms, where the entire input is known in advance, is fairly well understood: efficient polynomial time algorithms are known for the general and bipartite versions of the problem~\cite{DuanP14,HopcroftKarp1973,MuchaSankowski2004,Madry2013,Madry2016}. The online versions of the problem, where input is revealed one item at a time with the restriction of irrevocable decisions, have received a lot of attention in recent years due to applications in online advertising and algorithmic game theory (see the excellent survey by Mehta~\cite{mehta2013online} and references therein). The performance of an online algorithm $ALG$ is measured by its \emph{competitive ratio}, which is analogous to the notion of approximation ratio in the offline setting. The algorithm for a maximization problem is said to have \emph{asymptotic competitive ratio} at most $\rho$ if there is a constant $c \in \mathbb{R}$ such that for all inputs $I$ we have $ALG(I) \ge \rho \cdot OPT(I) + c$, where $OPT(I)$ denotes the value of the offline optimal solution on the input $I$. If the constant $c$ is zero then we say that the competitive ratio is \emph{strict}. The online maximum matching problem has been studied in a variety of settings, including adversarial~\cite{Karp1990Optimal,BirnbaumM2008,DevanurJK2013}, stochastic~\cite{bahmani2010improved,brubach2016new,Feldman2009Online,haeupler2011online,jaillet2014online,manshadi2012online}, and advice~\cite{durr2016power}. The advice setting is also the focus of the present paper, so we shall review it next. 

In the online advice model, the algorithm that receives input items one by one is cooperating with an all-powerful oracle, which has access to the entire input in advance. The oracle can communicate information to the algorithm by writing on an infinite tape, which it populates with an infinite binary string before the algorithm starts its execution. At any point during its runtime, the online algorithm can decide to read one or more bits from the tape. The worst-case number of bits read by the algorithm on an input of length $n$ is its advice complexity, as a function of input length $n$. The advice can be viewed as a generalization of randomness. A randomized algorithm is given random bits (independent of the input instance), whereas an advice algorithm is given advice bits prior to processing the online input. Note that unlike the standard Turing machine model with advice, in the online setting advice string contents are allowed to depend on the entire input and not just its length. Thus, advice length can be thought of as a measure of how much extra information about the input the online algorithm needs in order to achieve a certain competitive ratio. For more information about the online advice complexity see the excellent survey by Boyar et al.~\cite{boyar2016online}. This setting is important not only from the theoretical point of view but also from a practical one, as advice-based algorithms can often lead to efficient offline algorithms when advice is efficiently computable~\cite{borodin2019conceptually}. In addition, recently a new research direction has gained a lot of momentum, where the feasibility of using advice obtained by machine learning techniques in conjunction with online algorithms is being assessed (see~\cite{chen2021online,alomrani2021deep} for such results related to bipartite matching and \cite{angelopoulos2019online} references therein for a general theoretical framework).

In a variety of stochastic settings, algorithms for the general abstract online bipartite matching face the competitive ratio barrier of $1-1/e$. This is a consequence of a simple observation that when one throws $n$ balls into $n$ bins at random there will be roughly $n/e$ collisions, and this situation is often embedded in the online bipartite matching. In fact, the tight worst-case competitive ratio in the adversarial randomized setting is exactly $1-1/e$~\cite{Karp1990Optimal}. This competitive ratio is achieved by an algorithm that uses $\Theta(n \log n)$ random bits. Miyazaki~\cite{miyazaki2014advice} showed that in order to achieve optimality with advice, one needs $\Omega(n \log n)$ bits of advice. However, if one is satisfied with getting near optimality, D\"{u}rr et al.~\cite{durr2016power} showed that one can get the competitive ratio $(1-\epsilon)n$ using only $\Theta_\epsilon(n)$ bits of advice, where $\Theta_\epsilon$ hides constants depending on $\epsilon$.

Considerations in image processing~\cite{cohen1999finding} and circuit board design~\cite{hershberger1997efficient}, among other applications, motivate the study of matching problems in geometric graphs. In the geometric setting restricted to $2$ dimensions, one is led to consider various shapes in the Euclidean plane, such as circles, regular polygons, convex polygons, points, etc. Such shapes serve as vertices of the graph, on which the matching problem is defined. The adjacency can be defined by geometric considerations as well, such as two shapes intersecting each other, or shapes being reachable from one another via a curve in the plane that avoids all other shapes. The geometric setting motivates introducing new constraints to the problem such as non-crossing edges, which are important for the applications, such as circuit board design. The geometric settings are numerous and not as well understood as the general abstract setting (see the survey by Kano and Urrutia~\cite{kano2021discrete}), since introducing new constraints can change the complexity of the problem significantly resulting in some variants of the problem being $\mathcal{NP}$-hard~\cite{aloupis2013non}. As mentioned earlier, in this paper we study the geometric version of matching in $2$ dimensions, where geometric objects of interest are points and they can be matched with each other via non-crossing straight line segments. Observe that for our matching problems of interest we have $OPT(I) = |I|$ as discussed earlier, so the competitive ratio coincides with the fraction of matched points that the algorithm can guarantee in the worst case. The bichromatic version of the problem was first defined by Atallah who gave the $O(n \log^2 n)$ time deterministic offline algorithm for it~\cite{atallah1985matching}. Dumitrescu and Steiger~\cite{dumitrescu2000matching} generalized the problem and gave efficient approximation algorithms. The online versions of both monochromatic and bichromatic matching in the plane were introduced and studied by Bose et al.~\cite{bose2020non} and additional results are found in Sajadpour's Master's thesis~\cite{sajadpour2021non}. The main results from the work done on online MNM and BNM prior to our work can be summarized as follows (recall that the input length is $2n$ rather than $n$ in BNM and MNM):


\begin{itemize}
    \item a tight bound of $2/3$ on the deterministic competitive ratio for MNM, and a tight bound of $\log n - o(\log n)$ on the maximum worst-case number of matched points for BNM \cite{bose2020non};
    \item an upper bound of $2\log 3 n \approx 3.17 n$ and a lower bound of $\Omega(\log n)$ on the advice length to achieve optimality for MNM \cite{bose2020non};
    \item tight bound of $\Theta(n \log n)$ on the advice length to achieve optimality for BNM \cite{bose2020non} (note that the proof of the lower bound has a mistake);
    \item upper bound of $0.3304$ and lower bound of $0.0738$ on the competitive ratio of a randomized algorithm for MNM \cite{sajadpour2021non};
    \item lower bound of $\Omega(n)$ on the number of points that can be matched by a randomized algorithm in the worst case for BNM \cite{sajadpour2021non}.
\end{itemize}

Overall, the above set of results paints a picture that the bichromatic version is significantly more difficult than the monochromatic version of the problem, which is contrary to the offline setting, in which the bipartite setting is typically easier than the general setting. The above results indicate that the situation stays the same even in the presence of  advice. A significant special case of the BNM and MNM problems is when input points are located at the same distance from the origin, i.e., they are all located on a common circle. First of all, the known lower bounds are based on such inputs, and second of all, designing algorithms for this special case may be easier and often serves as the first step towards the algorithm for the general case. We refer to this special case as ``BNM/MNM on a circle'', and we say ``BNM/MNM on a plane'' to emphasize the general case without the common circle restriction. 

Our contributions are as follows. We observe that the lower bound of Bose et al. of $\Omega(n \log n)$ on the advice length to achieve optimality for BNM contains a bug. We give a new argument to establish the lower bound of $\log C_n \sim 2n - \frac{3}{2} \log n$, where $C_n$ is the $n^\text{th}$ Catalan number. The lower bound of Bose et al., as well as our new lower bound use input points that are located on a circle. We present an algorithm that uses $\log C_n$ bits of advice to solve BNM optimally on a circle. Thus, we completely resolve the advice complexity of BNM on a circle. We also show a lower bound of $n/3-1$ and an upper bound of $3n$ on the advice complexity for MNM on a plane. This gives an exponential improvement over the previously known lower bound and an improvement in the constant of the leading term in the upper bound. In addition, we establish
$\frac{\alpha}{2}\infdiv{\frac{2(1-\alpha)}{\alpha}}{1/4}n$ lower bound on the advice complexity to achieve competitive ratio $\alpha \in (16/17,1)$ for MNM on a circle where $\infdiv{p}{q}$ is the relative entropy between two Bernoulli random variables with parameters $p$ and $q$. The results are summarized in Table~\ref{tab:results}.

\begin{table}[h]
    \centering
    \begin{tabular}{|c||c|c||c|c|}
    \hline
        Problem Version&Previous LB & New LB  & Previous UB & New UB \\
        & &  (this work) & &  (this work)\\
        \hline
        MNM, Circle & $\Omega(\log n)$ & $n/3-1$ & $2n\log 3+o(n)$ & $\log C_n$\\
        \hline 
        MNM, Plane & $\Omega(\log n)$  & $n/3-1$ & $2n\log 3+o(n)$ &  $3n$\\
        \hline
        MNM, $\rho=\alpha$ & - & $\frac{\alpha}{2}\infdiv{\frac{2(1-\alpha)}{\alpha}}{1/4}n$ & -  & -\\
        \hline
        BNM, Circle & $n \log n$ (Mistake) & $\log C_n$& $n\log n$  & $\log C_n$\\
        \hline 
        BNM, Plane & $n \log n$ (Mistake) & $\log C_n$ & $n\log n$   &  - \\
        \hline
    \end{tabular}
    \caption{Summary of previously known results and our new results for the advice complexity of BNM and MNM.}
    \label{tab:results}
\end{table}

In terms of conceptual contributions, our work shows that the previously held belief that BNM is more difficult than MNM even in the presence of advice is no longer justified. Indeed, it might still turn out that BNM requires asymptotically more bits of advice than MNM to achieve optimality, but this would require new lower bound constructions and arguments where input points are not located on the same circle. Whether the advice complexity of BNM in the plane is $\omega(n)$ or $O(n)$ is left as an important open problem. Our lower bound arguments also provide a conceptual contribution to the advice complexity area, since MNM and BNM have an interesting phenomenon of complicated evolution of constraints during the runtime of the algorithm (because of the non-crossing condition). This makes it difficult (or perhaps even impossible) to use the standard tools from advice complexity, such as a reduction from the string guessing problem~\cite{bockenhauer2013} or partition trees~\cite{barhum2014power}. Thus, our lower bound arguments are based on the first principles. In the BNM case, we establish connections between BNM on the circle and a particular structured subset of the permutation group that is in one-to-one correspondence with full binary trees. Our further investigations of the BNM problem hint at deep connections between the theory of permutations and the BNM problem, which will perhaps be encountered when one would try to get tight bounds for the plane version of the problem. Our lower bound for the advice complexity of approximating MNM is based on probabilistic arguments and the nemesis input sequence is defined by a Markov chain that can successfully fool any online algorithm.

The rest of the paper is organized as follows. Preliminaries are presented in Section~\ref{sec:prelim}. We discuss the bug in the paper of Bose et al. in Section~\ref{sec:bichromatic}. We present our new lower bound argument and the upper bound for BNM on a circle in that section as well. In Section~\ref{sec:monochromatic} we present our results for the monochromatic version including the exponential improvement of the previously known lower bound, a slight improvement on the upper bound for achieving optimality, as well as a new lower bound to approximate MNM. Lastly, we conclude with some discussion and open problems in Section~\ref{sec:conclusion}.

\section{Preliminaries}
\label{sec:prelim}
\subsection{Online Non-Crossing Matching}\label{ssec:onm}
The input to MNM is a sequence of $2n$ points $p_1, p_2, \ldots, p_{2n} \in \mathbb{R}^2$ in general position arriving one-by-one. At step $i$, $p_i$ arrives and the algorithm can decide to match it with one of the unmatched points $p_j$ for some $j<i$ with a straight line segment or  leave it unmatched. The goal of the algorithm is to match all input points without creating any crossings of line segments.

In BNM, the first $n$ points, denoted by $B = (b_1,\ldots, b_n)$, are blue and can be thought of as given in advance. The next $n$ points, denoted by $R=(r_1,\ldots,r_n)$, are red and arrive one-by-one. The algorithm has to match red points (online) upon their arrival to the blue points (offline) without creating any crossings. Note that any algorithm for BNM can also be used for MNM by treating the first $n$ points $p_1, \ldots, p_n$ as blue and the next $n$ points $p_{n+1}, \ldots, p_{2n}$ as red. Consequently, an upper bound (on competitive ratio and/or advice complexity) for BNM implies the same upper bound for MNM, and a lower bound for MNM implies a lower bound for BNM. Many of our algorithms that achieve good performance on a circle also work on slightly more general inputs, namely, when input points are vertices of their convex hull. When this happens we say that points are in a ``convex position''.

Suppose that the input to MNM $I = (p_1, \ldots, p_{2n})$ is in convex position. Order the points in $I$ clockwise starting with $p_1$ in the convex hull of $I$. Let $j_i$ denote the index of point $p_i$ in this order with $j_1 = 0$. 
The \textit{parity} of $p_i$ is denoted by $\chi(p_i):= j_i \mod 2$. Therefore $\chi(p_1)=0$, the parity of clockwise neighbor of $p_1$ in the convex hull is $1$ and so forth. See Figure~\ref{fig:parity_example} for an illustration.

\begin{figure}[h]
    \centering
    \includegraphics[scale=0.8]{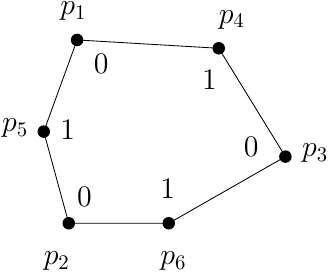}
    \caption{An instance of MNM in convex position with $6$ points with their convex hull and parities.}
    \label{fig:parity_example}
\end{figure}



\subsection{Online Algorithms with Advice}
\label{subsection:advice}
In this subsection, we briefly discuss two different models of algorithms with advice: the advice tape model and the parallel algorithms model. For completeness, we present a known argument demonstrating the equivalence (up to lower order terms) of two models. Thus, in the rest of the paper, we shall use two models interchangeably.

\textbf{Advice tape model.} 
In this model, there is an all-powerful oracle that sees the entire input in advance and populates an infinite advice tape with bits. Then the online algorithm $ALG$ receives the input one item at a time and $ALG$ has the ability to read any number of bits of advice from the tape at any point during the runtime. The oracle is trustworthy and always behaves in the best interests of the online algorithm according to a pre-agreed protocol. Thus, the oracle and the algorithm co-operate to achieve the best possible performance against an adversary. The worst-case number of advice bits read by the online algorithm with advice on inputs of length $n$ is the \emph{advice complexity} of the algorithm (as a function of $n$). The advice complexity of $ALG$ is denoted by $a(ALG, n)$. 

\textbf{Parallel online algorithms model.} In this model, a set of deterministic online algorithms (without advice) $\mathcal{A}$ is said to solve a set of inputs $\mathcal{I}$ if for every input sequence $I\in \mathcal{I}$ there is at least one algorithm $ALG\in \mathcal{A}$ such that $ALG$ solves $I$ optimally (or within the given level of approximation). We define the \emph{width} of $\mathcal{A}$ as $\lceil \log |\mathcal{A}| \rceil)$, denoted by $w(\mathcal{A})$.


\begin{fact}[Equivalence of the two models]
\label{fact:advice_models}
The advice tape model and the parallel online algorithms model are equivalent (up to additive $O(\log n)$ terms) in the following sense:

\begin{enumerate}
    \item If $ALG$ is an algorithm in the advice tape model that solves inputs of length $n$ correctly then there is a set $\mathcal{A}$ of algorithms in the parallel online algorithms model that solves the same inputs  correctly and such that $w(\mathcal{A}) \le a(ALG, n)$.
    \item If $\mathcal{A}$ is a set of algorithms that solves a set of inputs of length $n$ correctly in the parallel online algorithms model then there exists an algorithm $ALG$ in the tape model with advice that solves the same inputs correctly and such that $a(ALG, n) \le w(\mathcal{A}) + O(\log w(\mathcal{A}))$.
\end{enumerate}
\end{fact}
\begin{proof}
\hspace{1cm}

\begin{enumerate}
    \item For $x \in \{0,1\}^{a(ALG, n)}$ we define $ALG_x$ to be the algorithm $ALG$ with the first $a(ALG, n)$ bits on the advice tape fixed to be $x$ and restricted to inputs of length $n$. Then setting $\mathcal{A} = \{ALG_x\}_{x \in \{0,1\}^{a(ALG, n)}}$ establishes the claim.
    \item The oracle and the online algorithm $ALG$ agree on the ordering of $\mathcal{A}$. The oracle uses Elias delta coding scheme to encode the index of the first algorithm from $\mathcal{A}$ that solves the given input sequence correctly. This requires $w(\mathcal{A}) + O(\log w(\mathcal{A}))$ bits of advice, the rest of the bits are arbitrary. $ALG$ starts by reading and decoding the index of an algorithm from $\mathcal{A}$ written on the advice tape, and then runs that algorithm.
\end{enumerate}
\end{proof}

Since $O(\log w(\mathcal{A}))$ additive terms are small in our setting, we shall use the two models interchangeably and we shall sometimes refer to $w(\mathcal{A})$ as the advice complexity. Analogous statements and considerations hold for approximation.

\subsection{Catalan Numbers and Related Topics}
Catalan numbers \cite{oeis2021catalan}, denoted by $C_n$, are defined recursively as 
\[ C_n = \left\{
\begin{array}{ll}
1 & n = 0\\
\sum_{i=0}^{n-1} C_i C_{n-i-1} & n \ge 1
\end{array} \right..\]
The closed-form expression in terms of binomial coefficients for Catalan numbers is
$ C_n = \frac{1}{n+1} {2n \choose n}.$
By Stirling's approximation, we get 
$ \log C_n \sim 2n - \frac{3}{2} \log n.$ All logs are to the base $2$ unless stated otherwise.

The following three combinatorial objects are used in this paper extensively:
\begin{enumerate}
    \item A binary sequence $B=\{ b_1,...,b_{2n}\}$ is a \emph{Dyck word} if it has the equal number of $0$'s and $1$'s and in every prefix of it there are no more $1$'s  than $0$'s.
    \item A rooted ordered binary tree is a tree with a dedicated root vertex. Each vertex of such a tree has a left subtree and a right subtree (either of which can be empty). 
    \item A permutation of integers $\{1, 2, \ldots, n\}$ is a reordering of these integers. We represent a permutation $\sigma$ as a sequence $\sigma = (\sigma_1, \sigma_2, \ldots, \sigma_n)$ where $\sigma_i$ is the integer in the $i^\text{th}$ position according to $\sigma$. We say that $\sigma$ has a $231$ pattern if there exist  indices $i<j<k$ such that $\sigma_k<\sigma_i<\sigma_j$. If $\sigma$ does not have any $231$ pattern, we say $\sigma$ is \emph{$231$-avoiding}. For instance, $14235$ is $231$-avoiding but $31542$ is not since the subsequence $352$ is a $231$ pattern.
\end{enumerate}



Catalan numbers have been discovered in a variety of different contexts, so they have many alternative definitions. In particular, the above three definitions of combinatorial objects can be used to define Catalan numbers. An interested reader is referred to \cite{oeis2021catalan} and references therein.
\begin{fact}
\label{fact:catalan}
The number of Dyck words of length $2n$, the number of ordered rooted binary trees with $n$ vertices, and the number of $231$-avoiding permutations of $[n]$ are all equal to $C_n$.
\end{fact}

\section{Advice Complexity of Convex BNM}
\label{sec:bichromatic}

Bose et al. \cite{bose2020non} gave an argument that at least $\lceil \log n! \rceil \sim n \log n$ bits of advice are needed to solve BNM with $2n$ points optimally. Their argument is based on a family of input sequences in a convex position. Unfortunately, the proof contains a mistake and it is not possible to fix that mistake to restore the original bound, since there is an algorithm achieving asymptotically much better advice complexity of $\lceil \log C_n \rceil \sim 2n - \frac{3}{2} \log n$ for inputs in convex position. In this section, we begin by presenting this algorithm. The pseudocode is shown in the Algorithm \ref{alg:BTMatching}below. After that we show that $\lceil \log C_n \rceil$ bits of advice are necessary to achieve optimality for inputs on a circle, thus, establishing the tightness of $\lceil \log C_n \rceil$ bound. 
After presenting these arguments, we explain the mistake in the proof of Bose et al.


$BTMatching$ algorithm works on an input sequence of BNM as follows. Knowing the input sequence in which all blue points $B$ arrive before all red points $R$, the oracle finds a perfect non-crossing matching $M \subseteq R \times B$ between red and blue points. It is easy to see that such a matching always exists witnessed by a perfect matching of minimum total length, similarly to what we discussed in the introduction for the monochromatic case. The oracle then creates a tree representation $T$ of $M$ that takes into account the order in which red points arrive so that the online algorithm can later recover $M$ from $T$ on the fly. The tree is constructed recursively. Take the first red point $r_1$ and let $b_i$ denote the blue point to which $r_1$ is matched in $M$. The vector $b_i - r_1$ splits the plane into two half-planes: the left half-plane consists of all points $p$ such that $p - r_1$ forms an angle between $0$ and $\pi$ counter-clockwise with $b_i - r_1$, the right half-plane consists of all points $p$ such that $p - r_1$ forms an angle between $\pi$ and $2\pi$ counter-clockwise with $b_i - r_1$. Since the points are in convex position and $M$ is a non-crossing matching it follows that each line segment corresponds to matched  points in $M \setminus \{(r_1, b_i)\}$ lies entirely either in the left half-plane or the right-half plane. Thus, we partition $M = \{(r_1, b_i)\} \cup M_L \cup M_R$, where $M_L$ ($M_R$) consists of edges from $M$ that lie entirely in the left (respectively, right) half-plane. A tree $T_L$ ($T_R$) is constructed recursively for $M_L$ (respectively, $M_R$). Then $T$ is constructed by creating a root $t$ with $T_L$ as its left subtree and $T_R$ as its right subtree. Let $\langle T \rangle$ denote the encoding of a tree in binary. Then the oracle writes $\langle T \rangle$ on the tape. 

When the online algorithm is executed, it starts by reading $T$ from the tape. When a red point arrives the algorithm uses $T$ to deduce which blue point it should be matched to according to $M$. Let's first observe how it works for the first red point. When $r_1$ arrives the algorithm knows that this point corresponds to the root of $T$ and therefore it knows that the matching $M$ is such that there are $|M_L| = |T_L|$ edges to the left of the line segment by which $r_1$ is matched to a blue point. Thus, the algorithm can order blue points in clockwise order starting with $r_1$ and deduce that $b_i$ is the $|T_L|+1^\text{st}$ blue point in this order and match $r_1$ with $b_i$. After these points are matched, the problem splits into two independent problems corresponding to $T_L$ and $T_R$, which are located inside the two half-planes induced by the vector $b_i - r_1$, as described above. Thus, when the next red point $r_2$ arrives, the algorithm can determine whether $r_2$ lies in the left half-plane or the right-half plane and apply the same procedure as before with $T_L$ or $T_R$, respectively. And so on.
We use the following notation some of which depends on the current step in the execution of the algorithm: 
\vspace{-0.2cm}
\begin{itemize}
    \item $size(t)$ denotes the number of vertices in the subtree rooted at node $t$; $root(T)$ denotes the root of the (sub)tree $T$; $label(t)$ denotes the label of a node $t$; $left(t)$ ($right(t)$) denotes the left (respectively, right) child of the node $t$;
    \item $B_r$ is the set of currently available blue points such that matching the red point $r$ does not create any crossings with previously matched points;
    \item $C_r$ be the convex hull of the points in $B_r \cup \{r\}$;
    \item $b_r^i$ be the $i^\text{th}$ blue point in clockwise order in $C_r$.
\end{itemize}

\begin{algorithm}[h]
\caption{$BTMatchingOracle$}\label{alg:BTMatching}

\begin{algorithmic}

\Procedure{$BTMatchingOracle$}{$B=(b_1, \ldots, b_n), R = (r_1, \ldots, r_n)$}
\State{$M \gets $ non-crossing perfect matching $\subseteq R \times B$}
\State{$T \gets MatchingToBT(B,R,M)$}
\State{write $\langle T \rangle$ on the tape}
\EndProcedure

\end{algorithmic}
\end{algorithm}

\begin{algorithm}[h]
\caption{$MatchingToBT$}\label{alg:BTMatching2}
\begin{algorithmic}

\Procedure{$MatchingToBT$}{$B=(b_1, \ldots, b_n), R=(r_1, \ldots, r_n), M = \{e_1, \ldots, e_n\}$}
\If{$n=1$}
    \State{\textbf{return} $T$ consisting of a single node}
\EndIf
\State{let $i$ be such that $(r_1, b_i) \in M$}
\State{$B_L, R_L, M_L, B_R, R_R, M_R \gets \emptyset$}
\For{$j = 1$ to $n$}:
    \If{$j \neq i$}
        \If{$b_j$ is in the left half-plane induced by $(r_1,b_i)$}
            \State{$B_L.append(b_j)$}
        \Else
            \State{$B_R.append(b_j)$}
        \EndIf
    \EndIf
    \If{$j \neq 1$}
        \If{$r_j$ is in the left half-plane induced by $(r_1,b_i)$}
            \State{$R_L.append(r_j)$}
        \Else
            \State{$R_R.append(r_j)$}
        \EndIf
    \EndIf
    \If{$e_j \neq (r_1,b_i)$}
        \If{$e_j$ is in the left half-plane induced by $(r_1,b_i)$}
            \State{$M_L \gets M_L \cup \{e_j\}$}
        \Else
            \State{$M_R \gets M_R \cup \{e_j\}$}
        \EndIf
    \EndIf
\EndFor
\State {let $T_L=MatchingToBT(B_L, R_L, M_L)$}
\State {let $T_R=MatchingToBT(B_R, R_R, M_R)$}
\State {let $T$ be a tree with root $t$}
\State {$left(t) \gets root(T_L)$}
\State {$right(t) \gets root(T_R)$}
\State{\textbf{return} $T$}
\EndProcedure
\end{algorithmic}
\end{algorithm}


\begin{algorithm}[h]
\caption{$BTMatchingAlgorithm$}\label{alg:BTMatching3}

\begin{algorithmic}

\Procedure{$BTMatchingAlgorithm$}{}
\State{receive all blue points $B = (b_1, \ldots, b_n)$}
\State{read $T$ from the tape}
\For{every node $t \in T$}
    \State{$label(t) \gets \emptyset$}
\EndFor
\For{$i = 1$ to $n$}
    \State{receive $r_i$}
    \State{$t \gets root(T)$}
    \While{$label(t) \neq \emptyset$}
        \If{$r_i$ is on the left side of $label(t)$}
            \State{$t \gets left(t)$}
        \Else
            \State{$t \gets right(t)$}
        \EndIf
    \EndWhile
    \State{$k \gets size(left(t))+1$}
    \State{match $r$ to $b_r^k$}
    \State{$label(t) \gets (r, b_r^k)$}
\EndFor
\EndProcedure
\end{algorithmic}
\end{algorithm}

In the pseudocode, the procedure $BTMatchingOracle$ describes how the oracle operates. It uses $MatchingToBT$ subroutine to convert $M$ into a binary tree $T$. The procedure $BTMatchingAlgorithm$ describes how the online algorithm operates. As the online algorithm constructs the matching it labels the nodes by the edges of the matching constructed so far.

\begin{theorem}
$BTMatching$ (see Algorithms~\ref{alg:BTMatching}-\ref{alg:BTMatching3}) solves BNM with $n$ red and $n$ blue points in convex position optimally with $\lceil \log C_n \rceil$ bits of advice. 
\end{theorem}

\begin{proof} First, we justify the advice complexity of our algorithm. Let $\mathcal{T}_n$ be the set of all ordered unlabeled rooted binary trees with $n$ nodes. By Fact \ref{fact:catalan} we know $|\mathcal{T}_n|=C_n$. The oracle and the algorithm agree on an ordering of $\mathcal{T}_n$ and the bits $\langle T \rangle$ written by the oracle on the tape encode the index of $T$ according to this pre-agreed ordering. Thus the oracle can specify $T \in \mathcal{T}_n$ with $\lceil \log C_n \rceil$ bits of advice. Observe that the oracle does not need to use prefix-free code or specify $n$ separately, since the algorithm knows $n$ after receiving the blue points. Thus, the algorithm can deduce $C_n$ and read the first $\lceil \log C_n \rceil$ bits of advice from the tape after receiving the blue points. 
 
 From the description of the algorithm preceding the theorem, the correctness of the algorithm should be clear. We provide a brief argument by induction for completeness. Let $M'$ be the matching that is created by $BTMatching$ in the online phase. By induction on $n$, we prove that $M'$ is the same as the offline matching $M$.  The base case $n=1$ is trivial. Consider next some $n \ge 2$. Let $t$ be the root of $T$. Suppose $r_1$ is matched with $b_i$ in $M$. Let $B_L$ and $B_R$ be the blue points that are one left and right side of $e=(r_1,b_i)$ respectively and similarly define $R_L$ and $R_R$. Let $M_L$ be the matching between $B_L$ and $R_L$ and define $M_R$ accordingly. Since all the points are in convex position and edges of $M$ are non-crossing straight lines, there is no edge between the left and the right side of $(r_1,b_i)$, thus $M_L=B_L$ and the size of the left subtree of $t$ is $B_L$. In the online matching, $r_1$ will be matched to $b_{B_L+1}$ which is the same blue point as in $M$ since there are $B_L$ blue point on the left side of $(r_1,b_{B_L+1})$.

Let $T_L$ and $T_R$ be left and right subtrees of $t$ respectively. Subtrees $T_L$ and $T_R$ are created by edges in $M_L$ and $M_R$ respectively. Let $M'_L$ be the result of the algorithm with $B_L, R_L$ and $T_L$ as input. By induction, $M'_L = M_L$ and if we define $M'_R$ similarly, with the same argument $M'_R=M_R$. Note that $M'_L$ and $M'_R$ are obtained by splitting the input sequences into two and running the algorithm twice. In addition, whenever the online algorithm receives a red point on the left (right) side of $(r_1,b_i)$ it looks at $T_L$ ($T_R$). Thus red points on the left (right) side of $(r_1,b_i)$ will be matched with the same blue points as in $M'_L$ ($M'_R$) and we can conclude $M'=M$ which means the matching by the online algorithm is perfect and non-crossing.
\end{proof}


Next, we present a lower bound for BNM. Our lower bound uses input points located on a common circle, so it implies the same lower bound for BNM on inputs in convex position and in general position.

\begin{theorem}
Any online algorithm that solves BNM on a circle with $n$ red and $n$ blue points optimally uses at least $\lceil \log C_n \rceil$ bits of advice. 
\end{theorem}

\begin{proof}
We prove this by creating a family of input sequences $\mathcal{I}_n$ with $n$ red and $n$ blue points on a circle. 
In all input sequences in $\mathcal{I}_n$, blue points $B=(b_1,\ldots,b_n)$ have the exact same positions: they are located on the upper half of the circle and for each $1\leq i \leq n$, $b_i$ is the $i^\text{th}$ blue point from the left. Formally, 
$b_i$ has coordinates $(\cos{\alpha_i},\sin{\alpha_i})$ (Figure \ref{fig:BNM} left), where $\alpha_i= \pi(1- i/(n+1))$.

For each permutation $\sigma=(\sigma_1,\ldots,\sigma_n)$ on $[n]$, we generate the online input sequence $R(\sigma)=(r_1,\ldots,r_n)$ of red points in the lower half of the circle such that $r_i$ is the $\sigma_i^\text{th}$ red point from the left. Formally, let $\sigma_{-1}=0$ and $\sigma_0=n+1$ for the ease of the notation and auxiliary points $r_{-1}$ and $r_{0}$ on $(0,-1)$ and $(0,1)$ respectively. These two points are not in the input sequence and we define them for the ease of describing the input sequence. Next, we generate $(r_1, r_2, \ldots, r_n)$ in order. After generating $k$ points we have $k+1$ arcs between $r_{-1}$ and $r_0$ in counterclockwise order with points $r_{-1}, r_0, \ldots, r_k$ as boundary points between arcs. When we refer to $j^\text{th}$ arc, we mean $j^\text{th}$ arc in this order. To generate $r_i$ let $j = \min(i, \sigma_i)$ and place $r_i$ in the middle of $j^\text{th}$ arc. 
An example of such a construction for $R(2143)$ is illustrated in Figure~\ref{fig:BNM}: prior to $r_1$ being generated there is just one arc between $r_{-1}$ and $r_0$, so $j = \min(i, \sigma_i) = \min(1, 2) = 1$ and $r_1$ is placed in the middle of this one arc. This results in two arcs: one between $r_{-1}$ and $r_1$ and another between $r_1$ and $r_0$. For $i = 2$ we have $j = \min(2, \sigma_2) = (2, 1) = 1$ so we place $r_2$ in the middle of the first arc. For $i = 3$ we have $j = \min(3, \sigma_3) = \min(3,4) = 3$, so we place $r_3$ in the middle of the third arc, that is in the middle of the arc between $r_1$ and $r_0$. For $i = 4$ we have $\min(4, \sigma_4) = \min(4,3) = 3$, so we place $r_4$ in the middle of the third arc, that is in the middle of the arc between $r_1$ and $r_3$. Since input sequences in $\mathcal{I}_n$ only differ in red points we refer to an input sequence by its set of red points.


Let $\sigma$ and $\sigma'$ be two $231$-avoiding permutations over $[n]$ and $R(\sigma)=(r_1,\ldots, r_n)$ and $R(\sigma')=(r'_1,\ldots,r'_n)$ be their corresponding input sequences. Let $i$ be the first index that $\sigma$ and $\sigma'$ are different, i.e. $i = \argmin \{ j \in \{1, \ldots, n\} \mid \sigma_j \neq \sigma'_j \}$. It is clear from the above construction that the first $i-1$ points of $R(\sigma)$ and $R(\sigma')$ will be identical. The $231$-avoiding property implies that $r_i = r'_i$ as well. For $i=1$ the claim is trivial because the first red point in all $I\in \mathcal{I}_n$ goes to $(0,-1)$. For $i>1$, we establish the claim by contradiction. Assume that $r_i \neq r'_i$.  Prior to the placement of the $i^\text{th}$ points the two inputs partitioned the lower half of the circle into identical sequences of arcs. Since points are always placed in the middle of one of the arcs, it follows from $r_i \neq r'_i$ that the points $r_i, r'_i$ were placed in different arcs. Without loss of generality, suppose $r_i$ is placed in the middle of an arc to the right of $r'_i$ and there is a $1\leq j <i$ such that $r_j$ is between $r'_i$ and $r_i$. Since $\sigma_j=\sigma'_j$, the number of red points in $R(\sigma)$ to the left of $r_j$ is equal the number of red points in $R(\sigma')$ to the left of $r'_j$. Thus there should be a $i<k\leq n$ such that $r_k$ will be placed on the left side of $r_j$, i.e. $\sigma_k\leq \sigma_j\leq \sigma_i$. Hence the subsequence $\sigma_j,\sigma_i,\sigma_k$ is a $231$ pattern which is a contradiction.


Note that in an input sequence $I\in\mathcal{I}_n$, for $1\leq i\leq n$, if $r_i$ is the $j^\text{th}$ red point from left, it should be connected to $b_j$ in the unique perfect matching corresponding to $I$. The decision of a deterministic online algorithm is based on the prefix of the input sequence that it has received so far. With the same prefix of input sequences, $r_i$ and $r'_i$ should be matched to $b_{\sigma_i}$ and $b_{\sigma'_i} \neq b_{\sigma_i}$ respectively, therefore one deterministic algorithm cannot solve both $R(\sigma)$ and $R(\sigma')$ optimally and by pigeonhole principle together with Fact \ref{fact:catalan} we need at least $C_n$ deterministic algorithms and at least $\lceil\log C_n\rceil$ bits of advice for solving this problem optimally.
\end{proof}

\begin{figure}[h]
    \centering
    \includegraphics[scale=0.1]{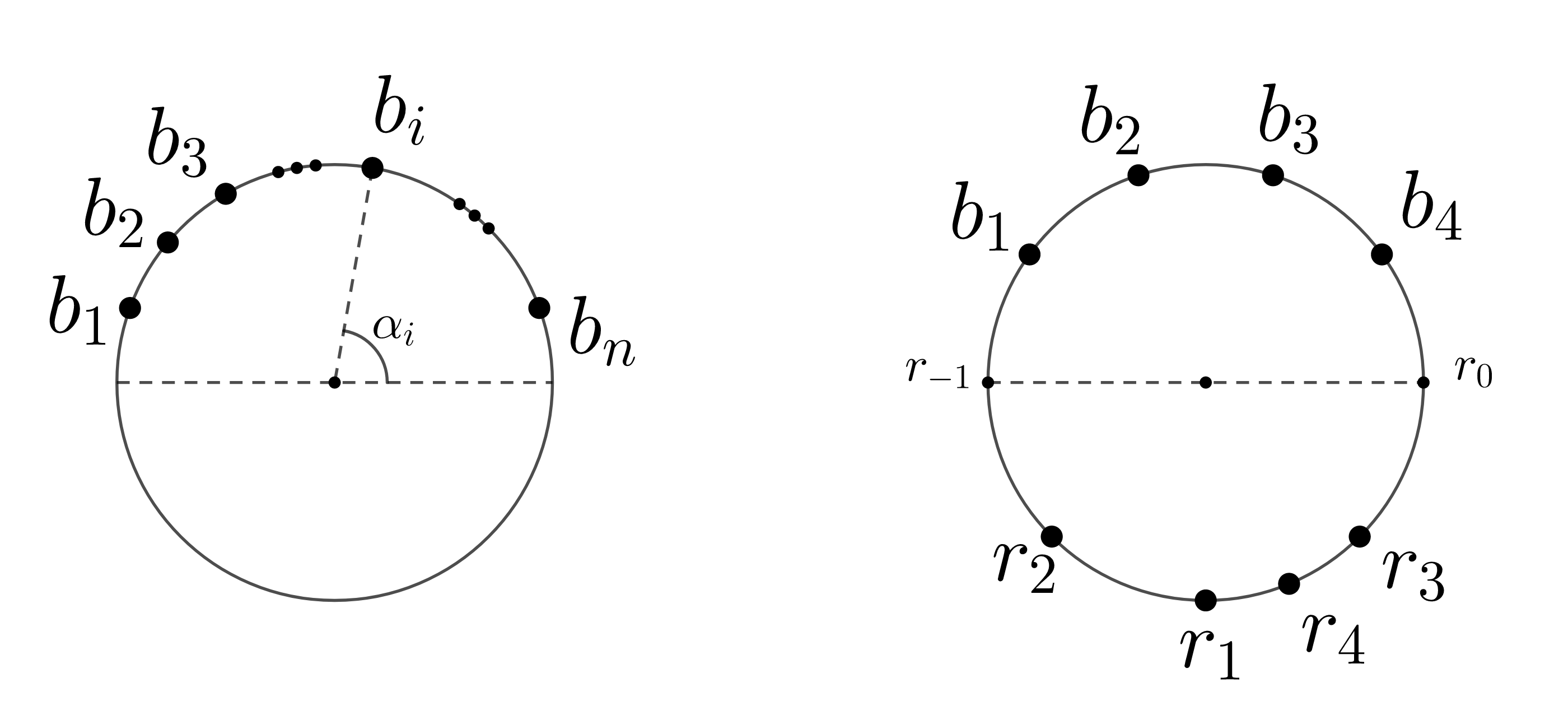}
    \caption{Left: positions of blue points. Right: the input sequence for $n=4$, corresponding the permutation $\sigma=(2,1,4,3)$.}
    \label{fig:BNM}
\end{figure}
Bose et al. \cite{bose2020non} had almost the same procedure of creating instances from permutations, but mistakenly, claimed no deterministic algorithm can solve two different permutations rather than two $231$-avoiding permutations. They did not consider that a single deterministic algorithm can solve inputs $R(\sigma)$ for multiple different $\sigma$ because of non-crossing constraints. For example, input sequences associated with permutations $\sigma=(2,3,1)$ and $\sigma'=(2,1,3)$ can be solved with one deterministic algorithm: the algorithm matches $r_1 = r'_1$ with $b_2$, then if $r_2$ arrives it will be in the right half-plane associated with $(r_1, b_2)$ edge, so it can be matched only with $b_3$ given the non-crossing constraint. If $r'_2$ arrives instead of $r_2$ then it will be in the left half-plane associated with $(r'_1, b_2)$ edge, so it can be matched only with $b_1$ given the non-crossing constraint. Thus, once the algorithm determines to match $r_1$ with $b_2$ it can complete it to a perfect matching in both $R(\sigma)$ and $R(\sigma')$. It implies that not every permutation out of $n!$ possibilities require a different deterministic algorithm. Another way of looking at it is that $r_i$ might be different from $r'_i$ where $i$ is the first coordinate at which $\sigma$ and $\sigma'$ differ unless $\sigma$ and $\sigma'$ are $231$-avoiding. If $r_i \neq r'_i$ then a single deterministic algorithm can use this information to match $r_i$ and $r'_i$ differently.

\section{Advice Complexity of MNM}
\label{sec:monochromatic}

In this section, we present our results for the advice complexity of MNM. We begin by presenting a bound of $3n$ on the advice complexity of solving MNM with general inputs in Subsection~\ref{ssec:ubgp}.
The algorithm using $\lceil \log C_n \rceil$ bits of advice for inputs in convex position for BNM from Section~\ref{sec:bichromatic} can be viewed as an algorithm for MNM that postpones matching points as long as possible leaving the first $n$ points unmatched (treating them as ``blue'' points). In Subsection~\ref{ssec:ubcp}, we present a family of algorithms that have an opposite behavior -- each algorithm in the family tries to match points as soon as possible, so we call them $ASAPMatching$ algorithms. We show that just like $BTMatching$ algorithm, $ASAPMatching$ algorithms achieve optimality with $\lceil \log C_n \rceil$ bits of advice for MNM on inputs in convex position. %
In Subsection~\ref{ssec:lbcirc} we present an $\lfloor n/3 \rfloor$ lower bound for optimally solving MNM on a circle, and in Subsection~\ref{ssec:lbapprox} we show an $\Omega_\alpha(n)$ lower bound for achieving competitive ratio $\alpha$.

\subsection{\texorpdfstring{$3n$}{3n} Upper Bound for General Position}
\label{ssec:ubgp}
Bose et. al. \cite{bose2020non} showed how to solve MNM with $(2 \log 3) n \approx 3.17 n$ bits of advice. In this section, we show that their result can be improved to $3n$ bits of advice by just using a slightly more efficient way of generating advice bits.

On input $P=(p_1,p_2,\ldots,p_{2n})$ the oracle sorts the points by their $x$-coordinates and matches consecutive pairs of points. Let $M$ be the resulting matching. Suppose $p_i$ is matched with $p_j$ in $M$, thus $p_j$ has the closest $x$-coordinate to $p_i$ either on its right or left. The advice string $A$ consists of $n$ parts, i.e. $A=a_1a_2\cdots a_{2n}$ such that for each $1\leq i \leq 2n$, $a_i$ is defined as follows:
\begin{itemize}
    \item $a_i=0$ if $j>i$, i.e. $p_j$ comes after $p_i$;
    \item $a_i=10$ if $i<j$ and $p_j$ is on the left side of $p_i$;
    \item $a_i=11$ if $i<j$ and $p_j$ is on the right side of $p_i$.
\end{itemize}

The online algorithm reads $A$ and reveals the sequence $a_1,\ldots,a_{2n}$ and for $1\leq i \leq 2n$ decides as follows:
\begin{itemize}
    \item if $a_i=0$, leaves $p_i$ unmatched;
    \item if $a_i=10$, matches $p_i$ with its closest $x$-coordinate on the its left;
    \item if $a_i=11$, matches $p_i$ with its closest $x$-coordinate on the its right.
\end{itemize}

Observe that this is just a prefix-free encoding/decoding of the three choices of an algorithm described above. The pseudocode is given in Algorithm~\ref{alg:sorted}.

\begin{algorithm}[h]
\caption{$SortedMatching$ algorithm.}\label{alg:sorted}
\begin{algorithmic}
\Procedure{$SortedMatchingOracle$}{$P=(p_1,\ldots,p_{2n})$}
\State{let $L$ be the sorted list of $P$ by their $x$-coordinates}
\State{$M\gets \emptyset$}
\For{$i=1$ to $n$}
    \State{let $p=L[2i+1]$ and $q=L[2i+2]$}
    \State{$M.append((p,q))$}
\EndFor
\State{$A \gets [ \ ]$}
\For{$i=1$ to $2n$}
    \State{find $j$ such that $(p_i,p_j)\in M$}
    \If {$j>i$}
        \State{$A.append(0)$}
    \ElsIf {$p_j$ is on the left of $p_i$}
            \State{$A.append(10)$}
    \Else
            \State{$A.append(11)$}
    \EndIf
\EndFor
\State{Write $A$ on the tape}
\EndProcedure\\

\Procedure{$SortedMatchingAlgorithm$}{}
\While{receive a new point $p_i$}
    \State{read a bit $b    _1$ from the tape}
    \If{$b_1 = 0$}
        \State{leave $p_i$ unmatched}
    \Else
        \State{read another bit $b_2$ from the tape}
        \If{$b_2=0$}
            \State{match $p_i$ with its closest $x$-coordinate point on its left}
        \Else
            \State{match $p_i$ with its closest $x$-coordinate point on its right}
        \EndIf
    \EndIf
\EndWhile
\EndProcedure
\end{algorithmic}
\end{algorithm}

\begin{theorem}
$SortedMatching$ (Algorithm \ref{alg:sorted}) solves MNM with $2n$ points in general position optimally with $3n$ bits of advice.
\end{theorem}

\begin{proof}
Let $(p_i,p_j)$ be an edge in the offline matching $M$ such that $i<j$. Since $p_i$ comes before $p_j$, $a_i$ is $0$ and $a_j$ is either $10$ or $11$ therefore for each edge in $M$ there are $3$ bits of advice and the size of the advice string is $3n$. Since the points are in general position, no three points may lie on the same vertical line. It follows that the matching $M$ produced by $SortedMatchingOracle$ is non-crossing and from the description of the algorithm it is easy to see that the online matching produced by $SortedMatchingAlgorithm$ is the same as $M$.
\end{proof}

\subsection{\texorpdfstring{$\lceil \log C_n \rceil$}{[log cn]} Upper Bound for Convex Position}\label{ssec:ubcp}

In this section, the set of \emph{available} points plays a crucial role. It is defined for an online algorithm as follows: when $p_i$ arrives, a point $p_j$ for $1\leq j < i$ is called \textit{available} if matching $p_i$ with $p_j$ does not create any crossing with the existing edges in the matching constructed by the algorithm by time $i$. Let $A_i$ denote the set of all available points for $p_i$. Observe that since the decisions of the online algorithm are deterministic, the oracle knows existing matching edges at time $i$, and hence it knows $A_i$. 

Recall that $\chi(p)$ denotes the parity of $p$ as described in Subsection~\ref{ssec:onm}. First, observe that in a perfect non-crossing matching $M$ if $p_i$ is matched with $p_j$ then the two points have opposite parities, for otherwise there would be an odd number of points in half-planes associated with the edge $(p_i, p_j)$, guaranteeing that at least one point on each side must remain unmatched. Our advice algorithm is based on the observation that this claim can be ``reversed'': as long as an online algorithm matches points of opposite parities without creating any crossings, this partial non-crossing matching remains valid, that is it can be extended to a perfect non-crossing matching of the whole instance. Of course, the parities are not known to the online algorithm and they cannot be inferred from the input seen so far, since they are based on a complete instance. One possibility is for an oracle to specify all parities with $2n$ bits of advice, but we can do it slightly more efficiently with $\lceil \log C_n \rceil$ bits of advice if $n$ is known to the algorithm in advance. If $n$ is not known then the savings from our more efficient encoding are diminished because of the need to encode numbers up to $C_n$ with a prefix-free code.

Observe that we can insist that as soon as $p_i$ arrives such that $A_i$ contains a point of opposite parity to $p_i$ then $p_i$ is matched. We refer to this property as ``as soon as possible'' or $ASAPMatching$, for short. In an $ASAPMatching$ algorithm, if one point in $A_i$ has parity opposite to $p_i$ then all points in $A_i$ have parity opposite to $p_i$. This means that as long as $ASAPMatching$ algorithm can infer that it is possible to match $p_i$ it can choose any point in $A_i$ for such matching. Therefore $ASAPMatching$ algorithms form a family of algorithms, where the specific algorithm is determined by how ties are broken. The tie-breaking rule is not important for our analysis, so we do not specify it explicitly.

Next, we describe the details of how our algorithm (see Algorithm \ref{alg:asap}) works. The oracle creates a binary string $D=(a_1,\ldots,a_{2n})$. For $i \in \{1, \ldots, 2n\}$, if there exists an available point $p_j \in A_i$ such that $\chi(p_j)\neq \chi(p_i)$  the oracle sets $a_i$ to $1$, otherwise it sets it to $0$. Thus, $a_i$ indicates whether $p_i$ could be matched at step $i$. The number of edges in a perfect matching is $n$ and for every $1\leq i\leq 2n$, at step $i$, the number of points matched at the time of their arrival is not more than the number of  points that were left unmatched at the time of their arrival hence $D$ is a Dyck word. Let $\langle D \rangle$ denote the encoding of a Dyck word in binary. If $n$ is known to the algorithm then $\langle D \rangle$ can consist of just $\lceil \log C_n \rceil$ bits. If $n$ is not known to the algorithm the oracle can encode $n$ using Elias delta coding followed by the encoding of $D$ as before, resulting in $\lceil \log C_n \rceil + \log n + O(\log \log n) \sim 2n -\frac{1}{2} \log n + O(\log \log n)$ bits. 
When the online algorithm is executed, it starts by reading $D$ from the tape. When an online point $p_i$ arrives, if $a_i=1$ the algorithm matches $p_i$ with one arbitrary point of $A_i$ and if $a_i=0$ it leaves $p_i$ unmatched. 


\begin{algorithm}[h]
\caption{$ASAPMatching$ algorithm.}\label{alg:asap}
\begin{algorithmic}
\Procedure{$ASAPMatchingOracle$}{}
\State{$D \gets [0]$}
\For{$i=2$ to $2n$}
    \If{there exist an available point $p_j$ for $p_i$ and $\chi(p_i)\neq \chi(p_j)$}
        \State{$D.append(1)$}
    \Else
        \State{$D.append(0)$}
    \EndIf
\EndFor
\State{write $\langle D \rangle$ on the tape}
\EndProcedure\\

\Procedure{$ASAPAlgorithm$}{}
\State{read $D$ from the tape}
\While{receive a new point $p_i$}
    \If{$d_i = 1$}
        \State{connect $p_i$ to one of the available points} 
    \Else
        \State{leave $p_i$ unmatched}
    \EndIf
\EndWhile
\EndProcedure
\end{algorithmic}
\end{algorithm}

\begin{theorem}
$ASAPMatching$ (see Algorithm \ref{alg:asap}) solves MNM with $2n$ points in convex position optimally with $\lceil \log C_n \rceil + \log n + O(\log \log n)$ bits of advice. Moreover, if $n$ is known to the algorithm only $\lceil \log C_n \rceil$ bits are required.
\end{theorem}

\begin{proof}
By Fact, \ref{fact:catalan}we know that the number of Dyck words with length $2n$ is $C_n$. Hence, the encoding described prior to this theorem achieves the desired advice complexity. 

$ASAPMatching$ matches $p_i$ with its available vertices thus it does not create crossing edges. It is left to see that the constructed matching is perfect. We can demonstrate it by strong induction on $n$. For $n = 1$ we have $\chi(p_1) \neq \chi(p_2)$ and we are done. For $n \ge 2$ define $i$ to be the smallest index such that $\chi(p_i) \neq \chi(p_1)$. $ASAPMatching$ then matches $p_i$ with $p_j$ for some $j \in \{1, \ldots, i-1\}$. The line passing through $p_i$ and $p_j$ splits the plane into two half-planes. Let $P_1$ consist of the points from the input that lie in one half-plane. Note that $P_1$ is a sequence and the order of points is the same as their order in $P$. We define $P_2$ similarly but for the other half-plane. Thus $P_1 = (p_{i_1}, p_{i_2}, \ldots, p_{i_k})$ and $P_2 = (p_{j_1}, p_{j_2}, \ldots, p_{j_m})$. Since every point other than $p_i$ and $p_j$ appears either in $P_1$ or $P_2$ we have $k + m + 2 = 2n$. Also, since $\chi(p_i) \neq \chi(p_j)$ we have that $P_1$ and $P_2$ contain even number of points each. Define $D_1 = (a_{i_1}, a_{i_2}, \ldots, a_{i_k})$ and $D_2 = (a_{j_1}, a_{j_2}, \ldots, a_{j_m})$ be the portions of the advice $D$ corresponding to items in $P_1$ and $P_2$, respectively. It is easy to see that $P_1$ forms an input to an MNM problem of size $k/2 < n$ and $D_1$ is a correct advice string corresponding to this input. Similarly for $P_2$. Thus, by induction assumption $ASAPMatching$ creates a perfect matching $M_1$ when it runs on $P_1$ with advice $D_1$, and it also creates a perfect matching $M_2$ when it runs on $P_2$ with $D_2$. Lastly, we observe that the decision of the algorithm for an input point $p$ depends on the advice bit and the set of available points. Since the advice bit and the set of available points for each $p \in P_1$ coincides with the advice bit and the set of available points for $p \in P$ after $p_i$ is matched with $p_j$, the algorithm will construct matching $M_1$ in $P$. Similarly, it will construct matching $M_2$ in $P$, and together with $p_i$ being matched with $p_j$ it gives a perfect matching for the entire instance $P$.

\end{proof}

\subsection{\texorpdfstring{$\lfloor n/3 \rfloor-1$}{[n/3]-1} Lower Bound}\label{ssec:lbcirc}

\begin{theorem}
Any online algorithm that solves MNM on a circle with $2n$ points optimally uses at least $\lfloor n/3 \rfloor-1$ bits of advice.
\end{theorem}
\begin{proof}
Let $n=3k$. We create a family of adversarial input instances $\mathcal{I}_n$ such that each instance $I\in \mathcal{I}_n$ contains $6k=2n$ points on the perimeter of a circle. The first $4k$ points $p_1, \ldots, p_{4k}$ are exactly the same in all instances $I \in \mathcal{I}_n$ and are positioned on the perimeter of a circle at regular angular intervals. These points arrive clockwise with $p_1$ located at the North pole.  The interval between $p_i$ and $p_{i+1}$ is called the $i^\text{th}$ interval (the $4k^\text{th}$ interval is between $p_{4k}$ and $p_1$).

Choose $j \in \{0, 1, 2, \ldots, 2k\}$ and choose a subset $S$ of $j$ intervals from the first $4k-1$ intervals. The next $j$ points are placed in the middle of each interval in $S$ and arrive in clockwise order. The remaining $2k-j$ points are placed in the $4k^\text{th}$ interval in clockwise order at regular angular intervals within the $4k^\text{th}$ interval. The family $\mathcal{I}_n$ arises out of all possible choices of $j$ and $S$.


\begin{equation}
\label{eq:size_of_family}
|\mathcal{I}_n| = \sum_{j=0}^{2k} \binom{4k-1}{j} \ge 2^{4k-2}+1.
\end{equation}

For each $I=(p_1,\ldots p_{6k})\in \mathcal{I}_n$, let $X(I)$ be the binary sequence $(\chi(p_1),\ldots, \chi(p_{4k}))$ of parities of the first $4k$ points of $I$ (recall that $\chi(p)$ is the parity of $p$ as described in Subsection~\ref{ssec:onm}). 
Consider  $I \neq I' \in \mathcal{I}_n$ and take the smallest index $j > 4k$ such that $p_j$ belongs to different intervals in $I$ and $I'$. Without loss of generality suppose that the location of $p_j$ in $I$ is before the location of $p_j$ in $I'$ in clockwise order from the North pole. Let $p_\ell$ be the clockwise neighbor of $p_j$ in $I$. Then the parity of $p_\ell$ in $I$ is different from the parity of $p_\ell$ in $I'$. This implies that $X(I) \neq X(I')$ demonstrating that $X : \mathcal{I}_n \rightarrow \{0,1\}^{4k}$ is one-to-one. 

If $M$ is a non-crossing matching on $p_1,\ldots,p_{4k}$ we call it a \textit{prior} matching. We say a prior matching $M$ is \textit{consistent} with an input sequence $I\in \mathcal{I}_n$ if it can be completed to a perfect non-crossing matching on points in $I$. The size of a consistent $M$ should be at least $k$ otherwise with more than $2k$ unmatched points and $2k$ arriving points it can not become a perfect matching. Moreover, for every $(p_i,p_j) \in M$, parities of $p_i$ and $p_j$ should be different, i.e. $\chi(p_i)\neq \chi(p_j)$. Otherwise $(p_i,p_j)$ splits the points of $I$ into two odd sets and it cannot become a perfect non-crossing matching. 


Since $X$ is one-to-one, a single prior matching $M$ can be consistent with at most $2^{3k}$ input sequences: $2^k$ different parities for points with opposite parities in $k$ matched edges and at most $2^{2k}$ different parities for other points.

The first $4k$ point of every in input sequence in $\mathcal{I}_n$ is the same, therefore for a deterministic algorithm, there is a one-to-one relation between the set of the prior matchings that it makes and the set of advice strings from the oracle for solving $\mathcal{I}_n$. Thus we say the advice string $A$ is consistent with input sequence $I\in \mathcal{I}_n$ if the prior matching that the algorithm makes with advice $A$ is consistent with $I$.

The size of the input family is greater than $2^{4k-2}$ and each advice string can be consistent with at most $2^{3k}$ input sequences therefore there should be more than $2^{k-2}$ different advice strings to cover all input sequences in $\mathcal{I}_n$ hence the algorithm needs at least $k-1$ bits of advice.

For $n=3k+1$ and $n=3k+2$, we create a family inputs the same way as $n=3k$ case but with $4k+1$ and $4k+2$ points in the fixed prefix respectively. With the same argument, we also need at least $k-1$ bits of advice in both cases.
\end{proof}

\subsection{\texorpdfstring{$\Omega_\alpha(n)$}{Oa(n)} Lower Bound for \texorpdfstring{$\alpha$}{a}-Approximation}\label{ssec:lbapprox}

In this subsection, we prove a lower bound on the amount of  advice needed for an online algorithm to match at least $2 \alpha n$ points without creating any crossing line segments. That is we study the advice complexity of achieving a strict competitive ratio $\alpha$. It is not clear whether it is possible to reduce from the binary string guessing problem~\cite{bockenhauer2014string} to our problem of interest. As such, instead of using the binary string guessing problem as a black box, we argue from first principles similar to the proof of the lower bound on string guessing. The argument is probabilistic: an algorithm that uses $k$ bits of advice and guarantees that at least $2\alpha n$ points are matched gives rise to a randomized algorithm that matches at least $2\alpha n$ points with a probability of at least $2^{-k}$ on every input (and consequently with respect to random inputs). This can be achieved by setting $k$ advice bits uniformly at random as opposed to having an oracle generate them. We exhibit a distribution on inputs of length $2n$ and show that a randomized algorithm cannot match a lot of points with high probability. The input is generated by a Markov chain process and all the points are on a circle. We maintain an active arc such that all future input points will be generated within this arc. The next point $p_i$ to arrive is placed in the middle of this arc. This splits the current arc into two and the process  continues on one of these two arcs chosen at random. However, it is not enough to simply continue the process on a randomly chosen arc, since an algorithm might decide to keep matching points as soon as they arrive and can be matched. Thus, we need to set up a probabilistic trap for this choice. We do this by deciding at random to insert a ``fake'' point in one arc and continue the process on another arc. We leave a possibility of not having a ``fake'' point and immediately continuing into one of the arcs chosen at random. This Markov chain process is designed so that with a constant probability we can guarantee a future unmatched point no matter whether $p_i$ is matched at time $i$ or is left unmatched by the algorithm. There are several possibilities to consider and the fact that the above trap works in all cases are a bit subtle. Nonetheless, the analysis only requires elementary probability theory and the following well-known tail bound:

\begin{fact}
\label{fact:binom-tail-bound}
If $X$ is a binomial random variable with variables $n$ and $p$, then by the Chernoff bound for $\alpha\in (0,p)$ we have:
$$P\{X\leq \alpha n\} \leq 2^{-n\infdiv{\alpha}{p}}$$
Where $\infdiv{\alpha}{p} = \alpha\log_2 \frac{\displaystyle \alpha}{\displaystyle p} + (1-\alpha)\log_2 \frac{\displaystyle 1-\alpha}{\displaystyle 1-p}$  is the \textit{relative entropy} between an $\alpha$-coin and a $p$-coin.
\end{fact}

Now, we are ready to present the main result of this subsection.

\begin{theorem}
Any online algorithm with advice for MNM that guarantees at most $2(1-\alpha)n$ unmatched points, where $\alpha \in (16/17,1)$, reads at least $\frac{\alpha}{2}\infdiv{\frac{2(1-\alpha)}{\alpha}}{1/4}n$ bits of advice.
\end{theorem}

\begin{proof}
We begin by describing a distribution of inputs consisting of $2n$ points on a circle, as discussed at the beginning of this section. Each point is either ``fake'' or ``parent'' -- terms that will become clear after we describe the distribution.

Create two sequences $F_1,\ldots, F_{2n}$ and $R_1, \ldots, R_{2n}$ of i.i.d. Bernoulli random variables with parameter $1/2$. Put points $p_1$ and $p_2$ on the North and the South poles of $\mathbb{S}^1$ ($(0,1)$ and $(0,-1)$) respectively and make $p_2$ a \textit{parent}. Starting from $i = 2$, follow the process: if $p_i$ is a parent, let $s_i^0$ and $s_i^1$ be its left and right adjacent arcs (maximal portions of the perimeter of a circle starting with $p_i$ and continuing until a previously generated point is encountered in clockwise and counterclockwise directions, respectively). Place $p_{i+1}$ in the middle of $s_i^{R_i}$. If $F_i=0$ make $p_{i+1}$ a parent and continue the process with $p_{i+1}$. If $F_i=1$, make $p_{i+1}$ a \textit{fake} point, place $p_{i+2}$ in the middle of $s_i^{1-R_i}$, make $p_{i+2}$ a parent  and continue with $p_{i+2}$.

From the above process we see that $F_i$ controls whether a point $p_i$ generated after a parent is fake, and $R_i$ controls whether $p_i$ is placed in the right arc or the left arc. We introduce $P_i$ as the indicator that $p_i$ is a parent. We have $P_i= 1-P_{i-1}F_i$. Next, we analyze the probability with which a randomized algorithm can achieve a large matching. Let $ALG$ be an arbitrary randomized algorithm. Note that at time $i$, $P_i, F_i$, and $R_i$ are not known to the algorithm.  Let $M$ be the size of the matching (random variable) constructed by $ALG$ and $T_1 < ... < T_M$ be the times at which the algorithm matched points. Let $A_i$ be the set of all available points to which $p_{T_i}$ can be connected. Suppose $p_{T_i}$ is parent and the algorithm matches it to $p_j\in A_i$, this matching splits $A_i\setminus\{p_j\}$ into $A_i^0$ and $A_i^1$, points on the left side and the right side of the matching respectively. 

A point becomes \text{isolated} if it is unmatched and cannot be matched afterward. If $|A_i^x|$ is zero and $p_{T_i+1}$ is fake and goes to $s_{T_i}^x$ it will be isolated. If $|A_i^x|$ is greater than zero and $p_{T_i+1}$ is a parent and it goes in $s_{T_i}^{1-x}$, points in $A_i^x$  become isolated.

For $1\leq i \leq M$ let $a_i\in \{00, 0+, +0, ++\}$ indicate the number of points in $A_i^0$ and $A_i^1$ at time $T_i$ (e.g. $a_i=+0$ indicates that $|A_i^0|>0$ and $|A_i^1|=0$). Let $X_i$ be the indicator that $p_{T_i}$ is a parent and the matching at time $T_i$ creates at least one isolated point. Then we have:
\[X_i = \left\{
\begin{array}{ll}
(1-P_{T_i-1}F_{T_i})F_{T_i+1} & \text{ if } a_i = 00 \\
(1-P_{T_i-1}F_{T_i})F_{T_i+1}(1-R_{T_i}) + (1-P_{T_i-1}F_{T_i})(1-F_{T_i+1})R_{T_i} & \text{ if } a_i = 0+ \\
(1-P_{T_i-1}F_{T_i})(1-F_{T_i+1})(1-R_{T_i}) + (1-P_{T_i-1}F_{T_i})F_{T_i+1}R_{T_i} & \text{ if } a_i = 0+ \\
(1-P_{T_i-1}F_{T_i})(1-F_{T_i+1}) & \text{ if } a_i = ++
\end{array} \right.
\]

Values of $X_1$ and $X_M$ do not follow this equation if $T_1=2$ or $T_M=2n$ but it does not affect the asymptotic result and we can ignore these scenarios. Now we define an auxiliary random sequence $Y_1,\ldots, Y_M$ as follows:
\[Y_i = \left\{
\begin{array}{ll}
(1-F_{T_i})F_{T_i+1} & \text{ if } a_i = 00 \\
(1-F_{T_i})F_{T_i+1}(1-R_{T_i}) + (1-F_{T_i})(1-F_{T_i+1})R_{T_i} & \text{ if } a_i = 0+ \\
(1-F_{T_i})(1-F_{T_i+1})(1-R_{T_i}) + (1-F_{T_i})F_{T_i+1}R_{T_i} & \text{ if } a_i = 0+ \\
(1-F_{T_i})(1-F_{T_i+1}) & \text{ if } a_i = ++
\end{array} \right.
\]

It is easy to check that $Y_i\leq X_i$ for $1\leq i\leq M$, for every sequence of decisions of an algorithm, and every outcome of $F_1, \ldots , F_M$ and $R_1,\ldots,R_M$. Note that $Y_2, Y_4, \ldots, Y_{\lfloor \frac{M}{2}\rfloor}$ are i.i.d. Bernoulli random variables with parameter $1/4$. Let $U$ be the number of unmatched points by $ALG$ then $U+2M=2n$ and $\sum_{i=1}^M X_i \leq U$ thus $\sum_{i=1}^{\lfloor \frac{M}{2}\rfloor} Y_{2i}\leq \sum_{i=1}^M X_i \leq 2n-2M$. We have:

\[P\{M \geq \alpha n\} = P\{\sum_{i=1}^M X_i \leq 2n-2M, M \geq \alpha n\} \leq P\{\sum_{i=1}^{\alpha n/2} Y_{2i} \leq 2(1-\alpha)n\}\]

For $\alpha\in (16/17,1)$ by Fact \ref{fact:binom-tail-bound}:


\[P\{M \geq \alpha n\} \leq 2^{-\frac{\alpha n}{2}\infdiv{\frac{4(1-\alpha)}{\alpha}}{1/4}} \]

To conclude the statement of the theorem, suppose we have an algorithm that guarantees at most $2(1-\alpha)n$ unmatched points with $f(\alpha)n$ bits of advice. If we replace advice bits with random bits and run this algorithm then with probability at least $2^{-f(\alpha)n}$ the algorithm creates at most $2(1-\alpha)n$ unmatched points. This implies that
$f(\alpha)\geq \frac{\alpha }{2}\infdiv{\frac{4(1-\alpha)}{\alpha}}{1/4}$.

\end{proof}

\section{Conclusion}
\label{sec:conclusion}

In this paper, we studied the advice complexity of online MNM and BNM. 
We showed that the advice complexity of solving BNM on a circle (or, more generally, on inputs in a convex position) is tightly bounded by the logarithm of the $n^\text{th}$ Catalan number from above and below. This result corrects the previous claim of Bose et al.~\cite{bose2020non} that the advice complexity is $\log (n!)$. At the heart of the result is a connection between non-crossing constraints in online inputs and the $231$-avoiding property of permutations of $n$ elements. The advice complexity of BNM on a plane is left as an open problem:

\begin{openproblem}
What is the advice complexity of BNM on a plane? 
\end{openproblem}

We know that the advice complexity of BNM lies between $\log C_n$ and $n \log n$. We suspect that if the answer to the above question is $\Theta(n)$, then the upper bound is due to a connection between non-crossing constraint in the plane and some structured subset of permutations, such as $\tau$-avoiding permutations for some pattern $\tau$. If the answer to the above question is $\omega(n)$ then this would demonstrate that BNM is, indeed, harder than MNM in the advice setting (which is provably so in the vanilla online model), but as of now, it remains unclear.  


Algorithms for BNM can be seen as a class of algorithms for MNM that wait for the first $n$ points to arrive and only then they start to match points. Since waiting longer would prevent an algorithm from creating a perfect matching, such algorithms postpone matching points as much as possible. On the other side of the spectrum, we introduced $ASAPMatching$ that matches points as soon as possible and we showed that this algorithm always solves MNM on a circle with at most $\log C_n$ bits of advice. 
It is interesting to investigate whether there are natural algorithms ``in between'' waiting for first $n$ arrivals versus matching as soon as possible and whether it makes any difference for MNM on a circle.

\begin{openproblem}
Can every  algorithm with advice for MNM on a circle be converted into a greedy algorithm with advice that matches points as soon as possible while preserving the advice complexity?
\end{openproblem}

We exponentially improved the lower bound on the advice complexity of solving MNM optimally from $O(\log n)$ (due to Bose et al.~\cite{bose2020non}) to $n/3-1$. We also established $\Omega_\alpha(n)$ lower bound for achieving competitive ratio $\alpha$. All our lower bounds are obtained on input points that are located on a common circle. The non-crossing constraint presented an obstacle to using standard proof techniques, such as a reduction from the string guessing problem. This motivates the following open problem.

\begin{openproblem}
Does there exist a reduction from the string guessing problem to MNM/BNM on a circle/plane?
\end{openproblem}

The approximation version of MNM/BNM considered in this paper is obtained by relaxing the perfect matching constraint and retaining the non-crossing constraint. It is quite natural to consider an alternative definition where the perfect matching constraint is retained and the non-crossing constraint is relaxed. That is, in this alternative approximation version of BNM/MNM an algorithm must always produce a perfect matching but the number of crossings of the edges of such a matching should be bounded by some parameter. Studying the advice complexity of this version and how it is related to the version of the problem in this work is another open problem of  interest.





\bibliographystyle{./styles/packages/unsrtabbrv}

\bibliography{bibs/full,bibs/refs}
 
\newpage

\end{document}